\documentclass[a4paper,10pt]{article}
\usepackage{times}
\usepackage{soul}
\usepackage{url}
\usepackage[utf8]{inputenc}
\usepackage[T1]{fontenc}
\usepackage{caption}
\usepackage{graphicx}
\usepackage{tikz, pgfplots}
\pgfplotsset{compat=1.18}
\usepackage{amsmath, amssymb}
\usepackage{varioref}
\usepackage{hyperref}
\usepackage{cleveref}
\usepackage{algorithm}
\usepackage{algorithmic}
\usepackage{lineno}
\usepackage{xcolor}
\usepackage{makecell}
\usepackage{mathtools,xspace,latexsym}
\usepackage{multicol,wrapfig,lipsum}
\usepackage{authblk}
\usepackage{geometry}
\newgeometry{
	left=1in,
	right=1in,
	top=1in,
	bottom=1.5in
}
\providecommand{\keywords}[1]
{
  \small	
  \textbf{\textit{Keywords:}} #1
}
\allowdisplaybreaks

\newcommand{\el}{\ensuremath{\ell}\xspace}

\newcommand{\suc}{\ensuremath{\succ}\xspace}
\newcommand{\FPT}{\ensuremath{\mathsf{FPT}}\xspace}

\newcommand{\NPH}{\ensuremath{\mathsf{NP}}-hard\xspace}

\newcommand{\DisKRA}{\textsc{Distinct Kemeny Ranking Aggregation}\xspace}
\newcommand{\DisKRAprime}{\textsc{Distinct Kemeny Ranking Aggregation}\ensuremath{^\prime}\xspace}
\newcommand{\DisOPTKRA}{\textsc{Distinct OPT Kemeny Ranking Aggregation}\xspace}
\newcommand{\DisApproxKRA}{\textsc{Distinct approximate Kemeny Ranking Aggregation}\xspace}

\newcommand{\NB}{\ensuremath{\mathbb N}\xspace}

\newcommand{\BB}{\ensuremath{\mathcal B}\xspace}
\newcommand{\CC}{\ensuremath{\mathcal C}\xspace}

\newcommand{\LL}{\ensuremath{\mathcal L}\xspace}

\newcommand{\OO}{\ensuremath{\mathcal O}\xspace}
\newcommand{\PP}{\ensuremath{\mathcal P}\xspace}

\newcommand{\RR}{\ensuremath{\mathcal R}\xspace}
\renewcommand{\SS}{\ensuremath{\mathcal S}\xspace}
\newcommand{\TT}{\ensuremath{\mathcal T}\xspace}

\newcommand{\XX}{\ensuremath{\mathcal X}\xspace}

\newtheorem{theorem}{Theorem}
\newtheorem{lemma}{Lemma}
\newtheorem{case}{Case}
\newtheorem{definition}{Definition}
\newtheorem{observation}{Observation}

\newtheorem{corollary}{Corollary}
\newenvironment{proof}{\noindent \textit{Proof.}}{\hfill$\square$}

\title{Parameterized Aspects of Distinct Kemeny Rank Aggregation}
\author[1]{Koustav De \\ email: \href{mailto:koustavde7@gmail.com}{koustavde7@gmail.com}}
\author[2]{\\Harshil Mittal \\ email: \href{mailto:mittal_harshil@iitgn.ac.in}{mittal\_harshil@iitgn.ac.in}}
\author[3]{\\Palash Dey \\ email: \href{mailto:palash.dey@cse.iitkgp.ac.in}{palash.dey@cse.iitkgp.ac.in}}
\author[4]{\\Neeldhara Misra \\ email: \href{mailto:neeldhara.m@iitgn.ac.in}{neeldhara.m@iitgn.ac.in}}
\affil[1,3]{Department of Computer Science and Engineering, Indian Institute of Technology Kharagpur}
\affil[2,4]{Department of Computer Science and Engineering, Indian Institute of Technology Gandhinagar}

\begin{document}
\maketitle
\begin{abstract}
The Kemeny method is one of the popular tools for rank aggregation. However, computing an optimal Kemeny ranking is \NPH. Consequently, the computational task of finding a Kemeny ranking has been studied under the lens of parameterized complexity with respect to many parameters. We first present a comprehensive relationship, both theoretical and empirical, among these parameters. Further, we study the problem of computing all distinct Kemeny rankings under the lens of parameterized complexity. We consider the target Kemeny score, number of candidates, average distance of input rankings, maximum range of any candidate, and unanimity width as our parameters. For all these parameters, we already have \FPT algorithms. We find that any desirable number of Kemeny rankings can also be found without substantial increase in running time. We also present \FPT approximation algorithms for Kemeny rank aggregation with respect to these parameters.
\end{abstract}
% !TeX root = main.tex

\section{Introduction}

Aggregating individual ranking over a set of alternatives into one societal ranking is a fundamental problem in social choice theory in particular and artificial intelligence in general. Immediate examples of such applications include aggregating the output of various search engines~\cite{DBLP:conf/www/DworkKNS01}, recommender systems~\cite{DBLP:conf/aaai/PennockHG00}, etc. The Kemeny rank aggregation method is often the method of choice in such applications due to its many desirable properties like Condorcet consistency that is electing the Condorcet winner (if it exists), etc. A Condorcet winner is a candidate who defeats every other candidate in pairwise election. The Kemeny method outputs a ranking $R$ with minimum sum of dissatisfaction of individual voters known as {\em Kemeny score} of $R$; the dissatisfaction of a voter with ranking $P$ with respect to $R$ is quantified as the number of pairs of candidates that $P$ and $R$ order differently~\cite{10.2307/20026529}. This quantity is also called the Kendall-Tau distance between $P$ and $R$. A ranking with minimum Kemeny score is called the Kemeny ranking.

The computational question of finding optimal Kemeny rankings is intractable in very restricted settings (for instance, even with a constant number of voters). Therefore, it has been well-studied from both approximation and parameterized perspectives. A problem is said to be \emph{fixed-parameter tractable} or \FPT with respect to a parameter $k$ if it admits an algorithm whose running time can be described as $f(k)\cdot n^{O(1)}$ where the input size is $n$, implying that the algorithm is efficient for instances where the parameter is ``small''~\cite{DBLP:books/sp/CyganFKLMPPS15}. For the Kemeny rank aggregation problem, the following parameters (among others) have enjoyed attention in the literature:

\begin{itemize}
    \item \emph{Range.} The range of a candidate in a profile is the difference between its positions in the voters who rank him/her the lowest and the highest~\cite{DBLP:journals/tcs/BetzlerFGNR09}. The maximum and average range of a profile is defined as, respectively, the maximum and average ranges of individual candidates. Profiles which are ``homogeneous'', $i.e.$ where most candidates are viewed somewhat similarly by the voters, are likely to have low values for range, while a single polarizing candidate can skew the max range parameter considerably.
    \item \emph{KT-distance.} The average (respectively, maximum) KT distance is the average (respectively, maximum) of the Kendall-Tau distances between all pairs of voters~\cite{DBLP:journals/tcs/BetzlerFGNR09}. Recall that the KT distance between a pair of rankings is the number of pairs that are ordered \emph{differently} by the two rankings under consideration. 
\end{itemize}

A pair of candidates are said to be \emph{unanimous} with respect to a voting profile if all votes rank them in the same relative order. Consider the following ``unanimity graph'' associated with a profile $P$ and defined as follows: every candidate is represented by a vertex, and there is an edge between a pair of candidates if and only if they are unanimous with respect to the profile. We use $G_P$ to denote this graph. Note that the structure of the \emph{complement} of this graph, denoted $\overline{G_P}$, carries information about candidates about whom the voters are not unanimous in their opinion. In particular, for every pair of candidates $a$ and $b$ that have an edge between them in the complement of the unanimity graph, there is at least one voter who prefers $a$ over $b$ and at least one who prefers $b$ over $a$. Thus every edge signals a lack of consensus, and one could think of the number of edges in this graph as a measure of the distance of the profile from an ``automatic consensus'', which is one that can be derived from the information about unanimous pairs alone. Motivated by this view, we propose and consider also the following structural parameters:

\begin{itemize}

\item \emph{Consensus Distance.} This is simply the number of edges in the complement of the unanimity graph $\overline{G_P}$.
\item \emph{Blocking Size.} This is the size of the largest \emph{clique} --- which is a collection of mutually adjacent vertices --- in the complement of the unanimity graph $\overline{G_P}$. It represents the largest number of candidates that the profile collectively finds mutually incomparable. 
\item \emph{Unanimity width.} It is the pathwidth of the unanimity graph $\overline{G_{P}} \text{ }i.e.$ the co-comparability graph or the complement of the comparability graph of the unanimity(partial) order of the input which is the specific order on the pairs of candidates on which all the voters agree, as studied by~\cite{DBLP:conf/ijcai/ArrighiFLO021} and it turns out to be a structural measure of how close the existing consensus in the input profile is to a complete ranking. The definition of pathwidth comes next.
\end{itemize}

The relationship between range and KT-distances is reasonably well understood, and these parameters are largely mutually incomparable. Our first contribution in this work is to extend these comparisons to the three parameters defined above, namely the consensus distance, blocking size, and unanimity width. The pathwidth of the complement of the unanimity graph turns out to be ``sandwiched'' between these two new parameters (consensus distance, blocking size) that we have proposed: it is at least the size of the largest independent set of the unanimity graph, and at most the number of edges in it. We compare these parameters and study them from an empirical perspective. We evaluate their values on various profiles sampled using a Mallows model on an assumed consensus.

% unanimity width parameter and two allied parameters that we believe are equally natural to consider: the size of a maximum independent set in the unanimity graph (this is the largest collection of mutually undecided pairs of candidates, and can be thought of as a reflection of how close the inherent consensus is to a ranking), and the number of non-edges in the unanimity graph (again, these are the number of pairs that are not determined in a obvious way by the profile, and can also be thought of as a measure of how ``far away'' the profile is from a consensus order). 

Our second contribution concerns enumerating optimal Kemeny rankings. In recent times, there is considerable research interest in finding a set of diverse optimal or near-optimal solutions of an optimization problem. Indeed, it is often difficult to encode all aspects of a complex system into a neat computational problem. In such scenarios, having a diverse set of optimal solutions for a problem $\Gamma$ allows the user to pick a solution which meets other aspects which are not captured in $\Gamma$. In the context of rank aggregation, such other external constraints may include gender fairness, demographic balance, etc. For the Kemeny rank aggregation method, Arrighi et al.~\cite{arrighi2020width} present a parameterized algorithm to output a set of diverse Kemeny rankings with respect to unanimity width as the parameter.

However, note that external requirements are often independent of the constraints in the optimization problem, and consequently they may not be correlated with diversity based on distance parameters. In particular, for useful externalities like gender fairness or geographic balance --- these features of the candidates may not have any relation with their position in the voters' rankings, and therefore, diversity \emph{between} solutions may not imply diversity \emph{within} any of the solutions. This becomes particularly stark when most near-optimal rankings do not meet the external requirements. Indeed, there is a substantial literature that considers the problem of accounting for these requirements explicitly, and studies trade-offs between optimality of solutions and the degree to which demands of diversity can be met.

In this contribution, we shift our focus from finding diverse solutions to finding as many \emph{distinct} solutions as possible. Enumerating solutions is a fundamental goal for any optimization problem. The literature on counting optimal Kemeny rankings is arguably limited considering that even finding one is hard in very restricted settings, and that instances could have exponentially many rankings --- which would be too expensive to enumerate. Indeed, consider a profile that consists of two votes over $m$ candidates, where one vote ranks the candidates in lexicographic order and the other ranks the candidates in reverse lexicographic order. For this instance, every ranking is an optimal ranking. However, note that real world preferences often have additional structure: for example, profiles with an odd number of voters that are single-peaked~\cite{DBLP:conf/ijcai/CornazGS13} or single-crossing~\cite{DBLP:conf/ijcai/CornazGS13} have unique optimal solutions. To address scenarios where the number of optimal solutions is large, we allow the user to specify the number $r$ of optimal solutions that she wants the algorithm to output. In our problem called \DisOPTKRA, the input is a set of rankings over a set of candidates and an integer $r$, and we need to output $\max\{r, \text{number of optimal solutions}\}$ Kemeny rankings.

% Hence, it makes limited sense in our opinion to try to meet these external constraints by outputting a set of diverse solutions, where diverseness is defined in terms of the parameters of the optimization problem. Of course, this solves the problem that some instance of typical optimization problems can have an exponentially many optimal solutions. However, we believe that such instances are rarity than norm. 

\subsection{Our Contributions}

\paragraph*{Experimental Results} We establish comprehensively relationships between all pairs of the following parameters: (a) maximum range, (b) average KT distance, (c) unanimity width, (d) blocking size, and (e) consensus distance. We also evaluate the values of these parameters on several profiles sampled using the Mallows model with various dispersion parameters. Intuitively, all of these parameters are proportional to the \emph{heterogeneity} of the profiles: in other words, more ``similar looking'' profiles have smaller values for these parameters. We are able to quantify this empirically by showing that the parameters decrease as we increase the dispersion of the Mallows distribution we sample from. It turns out that the higher the dispersion, the more the probability mass is concentrated around votes ``close to'' a central ranking. 

\paragraph*{Algorithms for Distinct Kemeny Rank Aggregation} The first parameter that we consider is the optimal Kemeny score $k$, also called the {\em standard parameter}. Many applications of rank aggregation, for example, faculty hiring, etc. exhibit correlation among the individual rankings --- everyone in the committee may tend to prefer some candidate with strong academic background than some other candidate with weak track record. In such applications, the optimal Kemeny score $k$, average Kendall-Tau distance $d$ (a.k.a. Bubble sort distance) among input rankings, maximum range of the positions of any candidate $r_{\text{max}}$, and unanimity width $w$ will be small, and an \FPT algorithm becomes useful. We show that there is an algorithm for \DisOPTKRA running in time $\OO^*\left(2^k\right)$~[\Cref{fpt:k}]. We next consider the number of candidates, $m$ as the parameter and present an algorithm running in time $\OO^*\left(2^m r^{\OO(1)}\right)$~[\Cref{fpt:m}] where $r$ is the required number of solutions. For $d$ and $r_{\text{max}}$, we present algorithms with running time $\OO^*(16^d)$ and $\OO^*\left(32^{r_{\text{max}}}\right)$~[\Cref{fpt:d,fpt:rmax}] respectively. Our last parameter is the unanimity width $w$ which is the pathwidth of the co-comparability graph of the unanimity order and we present an algorithm running in time $\mathcal{O}^{*}\left(2^{\mathcal{O}(w)}\cdot r\right)$~[\Cref{unanimitywidthalgo}].

Some instances may have a few optimal solutions, but have many close-to-optimal solutions. To address such cases, we study the \DisApproxKRA problem where the user gives a real number $\lambda\ge 1$ as input and looks for $\max\{r, \text{number of optimal solutions}\}$ rankings with Kemeny score at most $\lambda$ times the optimal Kemeny score. For this problem, we design algorithms with running time $\OO^*\left(2^{\lambda k}\right)$~[\Cref{cor:fpt-k}], $\OO^{*}\left( 2^{m}r^{\OO\left( 1 \right)} \right)$~[\Cref{corollary:Dis_Approx_fpt_m}] and $\OO^*\left(16^{\lambda d}\right)$~[\Cref{fpt:approx_d}].

We observe that the running time of all our algorithms are comparable with the respective parameterized algorithms for the problem of finding one Kemeny ranking. We note that this phenomenon is in sharp contrast with the diverse version of Kemeny rank aggregation where we have an \FPT algorithm only for unanimity width as the parameter. Also, the running time of the algorithm for the diverse version is significantly more than the standard non-diverse version~\cite{DBLP:conf/ijcai/ArrighiFLO021}. 

% \newpage

\subsection{Related Work}
Kemeny rule \cite{10.2307/20026529} shows us its most significant and popular mechanism for ranking aggregation. However, Bartholdi et al. \cite{bartholdi1989voting} have established that \textsc{Kemeny Score} is NP-complete even if we apply the restriction of having only four input rankings \cite{dwork2001rank}. Fixed-parameter algorithms for Kemeny voting rule have been proved to be an effective and important area for research by Betzler et al. \cite{DBLP:journals/tcs/BetzlerFGNR09} considering structural parameterizations such as ``number of candidates'', ``solution size \textit{i.e.} Kemeny Score'', ``average pairwise distance'', ``maximum range'', ``average range'' of candidates in an election. A multi-parametric algorithm for \textsc{Diverse Kemeny Rank Aggregation} over partially ordered votes has been studied in \cite{DBLP:conf/ijcai/ArrighiFLO021}. A small error in the construction proof from \cite{dwork2001rank} has been rectified by Biedl et al. \cite{biedl2009complexity}  and they have established the approximation factor of $2 - 2/k$,  improving from the previous approximation factor of 2. 

Further classification in more exact manner of the classical computational complexity of Kemeny elections has been provided by Hemaspaandra et al. \cite{hemaspaandra2005complexity}. With respect to the practical relevance of the computational hardness of the \textsc{Kemeny Score}, polynomial-time approximation algorithms have been developed where a factor of $8 / 5$ is seen in \cite{van2009deterministic} and a factor of $11 / 7$ is proved in \cite{ailon2008aggregating}. As we can see a polynomial-time approximation scheme (PTAS) developed by Kenyon-Mathieu and Schudy \cite{kenyon2007rank} is not very practical, we can refer to various approximation algorithms and heuristics provided by Schalekamp and van Zuylen \cite{schalekamp2009rank} as an evaluation. Some greedy techniques and branch-and-bound methods can be seen from Conitzer, Davenport and Kalagnanam \cite{davenport2004computational,conitzer2006improved} where the authors have performed their studies heuristically. From \cite{DBLP:conf/aaai/PennockHG00,DBLP:conf/www/DworkKNS01} we can verify the methods used for merging results from various search engines and the notion of collaborative filtering. 

Polynomial time algorithms producing good solutions for rank aggregation rule is a consequence of thorough computational studies \cite{DBLP:conf/waoa/ZuylenW07,DBLP:journals/jacm/AilonCN08}. Cornaz et al. \cite{DBLP:conf/ijcai/CornazGS13} have established polynomial time computability of the single-peaked and single-crossing widths and have proposed new fixed-parameter tractability results for the computation of an optimal ranking according to the Kemeny rule by following the results of Guo te al. \cite{DBLP:conf/iwpec/GuoHN04}. In social choice theory \cite{klamler2004dodgson,bartholdi1989voting}, the ideas related to diverse sets of solutions have found tremendous applicability. The study in \cite{zwicker2018cycles} introduced the $(j,k)$-Kemeny rule which is a generalization of Kemeny's voting rule that aggregates ballots containing weak orders with $j$ indifference classes into a weak order with $k$ indifference classes. In social choice theory, different values of $j$ and $k$ yield various rules of the interest of the community turning up as special cases. The minimum Kendall-Tau distance between pairs of solutions has a nice analogy with $min$ Hamming distance over all pairs of solutions as shown in \cite{hebrard2005finding,hebrard2007distance}.
\section{Preliminaries}
For an integer $\el$, we denote the set $\{ 1,\ldots,\el \}$ by $\left[ \el \right]$. For two integers $ a,b $, we denote the set $\{ i \in \NB : a \le i \le b \}$ by $\left[ a,b \right]$. Given two integer tuples $\left( x_1, \ldots , x_\el \right), \left( y_1 , \ldots , y_\el \right) \in \NB^\el $, we say $\left( x_1, \ldots, x_\el \right) >_{ \text{lex} } \left( y_1, \ldots, y_\el \right)$ if there exists an integer $i \in \left[ \el \right] $ such that we have (i) $x_j=y_j$ for every $ j \in \left[ i-1 \right]$, and (ii) $x_i > y_i$.

Let $\CC$ be a set of candidates and $R=\left\{\pi_{1}, \ldots, \pi_{r}\right\}$ a multi-set of rankings (complete orders) on $\CC$. For a ranking $\pi$ and a candidate $c$ let us define $\text{pos}_\pi(c)$ to be $|\left\{ c^{\prime} \in \CC : c^{\prime} \succ_{\pi} c\right\}|$. We define the {\em range} $r$ of $c$ in a set of rankings $\Pi$ to be $\max\limits_{\pi_{i}, \pi_{j} \in \Pi} \left\{ |\text{pos}_{\pi_{i}}\left( c \right) - \text{pos}_{\pi_{j}}\left( c \right)| \right\} + 1$. We denote the set of all complete orders over $\CC$ by $\LL(\CC)$. The Kemeny score of a ranking $Q\in\LL(\CC)$ with respect to $R$ is defined as
\[ \text{Kemeny}_R(Q) = \sum_{i=1}^r \text{d}_{\text{KT}} (Q, \pi_i)\]
where d$_\text{KT}(\cdot,\cdot)$ is the Kendall-Tau distance -- the number of pairs of candidates whom the linear orders order differently -- between two linear orders, and $N_R(x>y)$ is the number of linear orders in $R$ where $x$ is preferred over $y$. A Kemeny ranking of $R$ is a ranking $Q$ which has the minimum $\text{Kemeny}_R(Q)$; the score $\text{Kemeny}_R(Q)$ is called the optimal Kemeny score of $R$.

We now define our problems formally. For a set of rankings $\Pi$, we denote the set of (optimal) Kemeny rankings and rankings with Kemeny score at most some integer $k$ for $\Pi$ respectively by $K(\Pi)$ and $K(\Pi,k)$, and the minimum Kemeny score by $k_{\text{OPT}}(\Pi)$.

\begin{definition}(\DisOPTKRA).
Given a set of rankings (complete orders) $\Pi$ over a set of candidates $\CC$ and integer $r$, compute $\el=\min\{r,|K(\Pi)|\}$ distinct Kemeny rankings $\pi_{1}, \ldots, \pi_{\el}$. We denote an arbitrary instance of it by $(\CC,\Pi,r)$.
\end{definition}

For a set of rankings $\Pi$ over a set of candidates $\CC$, we say that a complete order $\pi$ respects unanimity order if we have $x\suc_\pi y$ whenever $x\suc y$ for all $\suc\in\Pi$.

\begin{definition}(\DisApproxKRA).
Given a set of ranking (complete order) $\Pi$ over a set of candidates $\CC$, an approximation factor $\lambda\ge 1$, and integer $r$, compute $\el=\min\{r,|K(\Pi,\lambda\cdot k_{\text{OPT}}(\Pi))|\}$ distinct rankings $\pi_{1}, \ldots, \pi_{\el}$ such that each ranking $\pi_i, i\in[\el]$ respects unanimity order with respect to $\Pi$ and the Kemeny score of each ranking $\pi_i, i\in[\el]$ is at most $\lambda \cdot k_{\text{OPT}}(\Pi)$. We denote an arbitrary instance of it by $(\CC,\Pi,\lambda,r)$.
\end{definition}

\begin{definition}(\DisKRA).
Given a list of partial votes $\Pi$ over a set of candidates $\CC$, and integers $k$ and $r$, compute $\el=\min\{r,|K(\Pi,k)|\}$ distinct rankings $\pi_{1}, \ldots, \pi_{\el}$ such that the Kemeny score for each ranking $\pi_{i}$ is at most $k$ and each $\pi_i, i\in [\el]$ respects unanimity order? We denote an arbitrary instance of it by $(\CC,\Pi,k,r)$.
\end{definition}

We use $\OO^*(\cdot)$ to hide polynomial factors. That is, we denote $\OO(f(k)\text{poly}(n))$ as $\OO^*(f(k))$ where $n$ is the input size.

We define a path decomposition of a graph $G = \left(V, E \right)$ by a tuple $\PP = \left( \BB_{i}\right)_{i \in \left[ t  \right]}$ where each bag $\BB_{i} \subseteq V$, $t$ is the number of bags in $\PP$ and $\PP$ satisfies the following additional constraints : (1) $\bigcup_{i\in \left[ t \right]} \BB_{i} = V$, (2) $\exists i \in \left[ t \right]$ such that $u, v \in \BB_{i} ~\text{for each } \left( u, v \right) \in E$ and (3) $\BB_{i} \cap \BB_{k} \subseteq \BB_{j}$ for each $i, j, k \in \left[ t \right]$ satisfying $i < j < k$. The width of $\PP$ denoted by $w\left( \PP \right)$ is defined as $\max_{i \in \left[ t \right]} |\BB_{i}| - 1$. The $pathwidth$ of $G$ is denoted by $pw\left( G \right)$ which is defined as the minimum width of a path decomposition of $G$.

For a dispersion parameter $\theta> 0$ and a central ordering $\pi_0$ of a set  $\CC$ of candidates, the \emph{Mallows model} associates the following probability mass with each ordering $\pi$ of $\CC$:
$$\frac{e^{-\theta\cdot \big(\mbox{d}_{\text{KT}}(\pi,\pi_0)\big)}}{\underset{\tiny{\substack{\sigma: \sigma \in \LL(\CC)}}}{\Sigma}e^{-\theta\cdot \big(\mbox{d}_{\text{KT}}(\sigma,\pi_0)\big)}}.$$

\section{Parameters and Experiments}
\def\arraystretch{1.25}
\begin{table*}[htbp]
\begin{center}
\begin{tabular}{llllll}
                      & range & average KT distance & unanimity width   & blocking size     & distance to consensus \\ \hline
range                 &       & $\star$ $\dagger$   & $\geq$ $\dagger$  & $\star$ $\dagger$ & $\leq$ $\star$        \\ \hline
average KT distance   &       &                     & $\star$ $\dagger$ & $\star$ $\dagger$ & $\leq$ $\star$        \\ \hline
unanimity width       &       &                     &                   & $\geq$ $\dagger$  & $\leq$ $\star$        \\ \hline
blocking size         &       &                     &                   &                   & $\leq$ $\star$        \\ \hline
distance to consensus &       &                     &                   &                   &                       \\ \hline
\end{tabular}
\caption{Relationships between parameters. Read the entries in the table as follows: the row label, followed by the entry, followed by the column label. A ``$\star$'' is to be read as ``can be arbitrarily smaller than''; while a ``$\dagger$'' is to be read as ``can be arbitrarily larger than''. The signs are to be read as is, but is a slight abuse of notation in the sense that there may be constant factors involved in the inequalities.}
\label{tab:param-compare-1}
\end{center}
\end{table*}
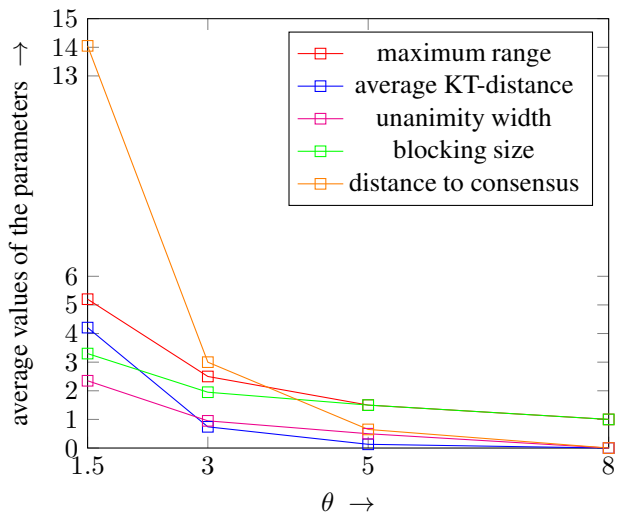
\begin{figure}[htbp]
\begin{tikzpicture}
\begin{axis}[
    title={},
    xlabel={$\theta~\rightarrow$},
    ylabel={average values of the parameters$~\rightarrow$},
    xmin=1.5, xmax=8,
    ymin=0, ymax=15,
    xtick={1.5,3,5,8},
    ytick={0,1,2,3,4,5,6,13,14,15},
    legend pos=north east,
    %ymajorgrids=true,
    %grid style=dashed,
]

\addplot[
    color=red,
    mark=square,
    ]
    coordinates {
    (1.5,5.200)(3,2.500)(5,1.500)(8,1.000)
    };
    \addlegendentry{maximum range}

\addplot[
    color=blue,
    mark=square,
    ]
    coordinates {
    (1.5,4.207)(3,0.741)(5,0.130)(8,0.000)
    };
    \addlegendentry{average KT-distance}

\addplot[
    color=magenta,
    mark=square,
    ]
    coordinates {
    (1.5,2.350)(3,0.950)(5,0.500)(8,0.000)
    };
    \addlegendentry{unanimity width}

\addplot[
    color=green,
    mark=square,
    ]
    coordinates {
    (1.5,3.300)(3,1.950)(5,1.500)(8,1.000)
    };
    \addlegendentry{blocking size}

\addplot[
    color=orange,
    mark=square,
    ]
    coordinates {
    (1.5,14.050)(3,3.000)(5,0.650)(8,0.000)
    };
    \addlegendentry{distance to consensus}
\end{axis}
\end{tikzpicture}
\caption{This plot illustrates that for $200$ votes and $10$ candidates, the average (over $20$ samples) values of the parameters drop with increase in $\theta$.}
\label{param-compare-1}
\end{figure}

The parameter comparisons are summarized in \Cref{tab:param-compare-1}. Here we briefly justify the claims implicit to \Cref{tab:param-compare-1}, noting that the incomparability of range and average KT distance was established already in~\cite{DBLP:journals/tcs/BetzlerFGNR09}, and we do not repeat those examples here. The following constructions imply the remaining claims.

 Let $p$ be an integer. Consider a profile with $2p$ candidates, say $c_1, \ldots,c_p, d_1, \ldots,d_p$. We have two votes: $c_1> \cdots >c_p>d_1> \cdots >d_p$, and $d_1> \cdots >d_p>c_1> \cdots >c_p$. Note that the complement of the unanimity graph is a complete bipartite graph, with each part having size $p$. Here, the unanimity width is $p$, the range is $p$, the distance to consensus is $p^2$, but the blocking size is $2$. This shows that the blocking size can be arbitrarily smaller than these remaining parameters. 

Consider a profile with two votes $a>b>c>d>\cdots$ and $b>a>d>c>\cdots$; i.e, where consecutive candidates are swapped. The complement of the unanimity graph is a matching on $m/2$ edges, the average KT distance is $m/2$, while the range is $2$ and the unanimity width is $1$.

Consider a profile where a particular vote is repeated $(n-1)$ times and the last vote is the reversal of the common vote. This profile has constant average KT distance, but all remaining parameters are functions of $m$. 

Next, we turn to the bounds. Note that the claim that the consensus distance is an upper bound for the blocking size and unanimity width follows directly from graph-theoretic definitions. Further, the claim that is an upper bound for the average KT distance follows from the fact that the average KT distance is at least the minimum KT distance, which is witnessed by some specific pair of votes: but these manifest directly as edges in the complement of the unanimity graph. 

The claim that the distance to consensus is an upper bound for the maximum range follows from the fact that if the maximum range is $r$, then there are at least $r$ non-unanimous pairs in the profile, each of which contributes an edge to the complement of the unanimity graph. The fact that the unanimity width is at least the blocking size is also a standard graph-theoretic fact (the pathwidth is lower bounded by the clique size). The fact that the pathwidth is upper bounded by the range can be observed by constructing an appropriate path decomposition over $m$ bags, where the $i$-th bag contains all candidates whose range contains the $i$-th position. It is easily verified that this is a valid path decomposition whose width is $O(r)$.

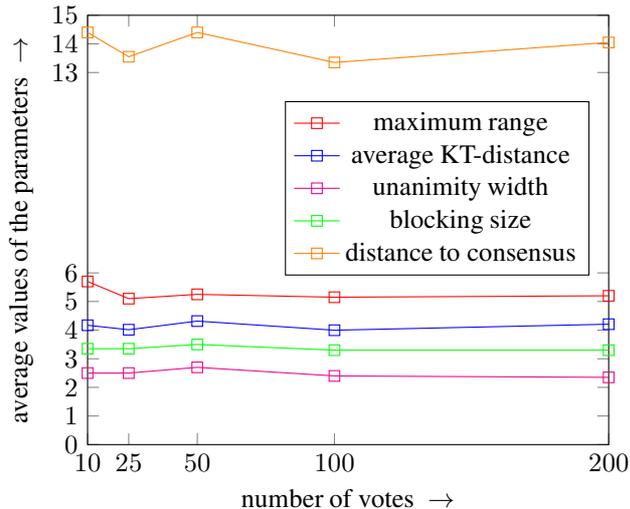
\begin{figure}
\begin{tikzpicture}
\begin{axis}[
    title={},
    xlabel={number of votes$~\rightarrow$},
    ylabel={average values of the parameters$~\rightarrow$},
    xmin=10, xmax=200,
    ymin=0, ymax=15,
    xtick={10,25,50,100,200},
    ytick={0,1,2,3,4,5,6,13,14,15},
    %legend pos= north east
    legend style={at={(axis cs:193,12)},anchor=north east}
    %ymajorgrids=true,
    %grid style=dashed,
]

\addplot[
    color=red,
    mark=square,
    ]
    coordinates {
    (10,5.700)(25,5.100)(50,5.250)(100,5.150)(200,5.200)
    };
    \addlegendentry{maximum range}

\addplot[
    color=blue,
    mark=square,
    ]
    coordinates {
    (10,4.170)(25,4.014)(50,4.316)(100,3.997)(200,4.207)
    };
    \addlegendentry{average KT-distance}

\addplot[
    color=magenta,
    mark=square,
    ]
    coordinates {
    (10,2.500)(25,2.500)(50,2.700)(100,2.400)(200,2.350)
    };
    \addlegendentry{unanimity width}

\addplot[
    color=green,
    mark=square,
    ]
    coordinates {
    (10,3.350)(25,3.350)(50,3.500)(100,3.300)(200,3.300)
    };
    \addlegendentry{blocking size}

\addplot[
    color=orange,
    mark=square,
    ]
    coordinates {
    (10,14.400)(25,13.550)(50,14.400)(100,13.350)(200,14.050)
    };
    \addlegendentry{distance to consensus}
\end{axis}
\end{tikzpicture}
\caption{This plot illustrates that for $\theta=1.5$ and $10$ candidates, the average (over $20$ samples) values of the parameters do not change much with increase in the number of votes.}
\label{param-compare-2}
\end{figure}

\paragraph*{Experimental Setup.} We computed the values of the five parameters, namely (a) maximum range, (b) average KT distance, (c) unanimity width, (d) blocking size, and (e) consensus distance, on profiles with 10 candidates and 10, 25, 50, 100, and 200 voters. We used two inbuilt functions of SageMath, namely clique\_number() and treewidth(), to compute blocking size and unanimity width respectively. The remaining parameters were computed with a direct implementation.

Each value reported is averaged over 20 samples. We also varied the dispersion parameter between the values 1.5, 3, 5, and 8. Our main observation from the empirical data was twofold: first, the values of all parameters dropped as we increased the dispersion parameter, which is as one would expect, since a higher dispersion parameter gives us more homogeneous profiles (cf. \Cref{param-compare-1}); and second, for a fixed dispersion parameter, the values of the parameters did not change much between small and large profiles (i.e, variations in the number of votes did not lead to large variations in the parameter, cf. \Cref{param-compare-2}). 

\def\arraystretch{1.25}
\begin{table}[htbp]
\begin{tabular}{llllll}
\#Voters            & 10 & 25 & 50 & 100 & 200 \\ \hline
maximum range                 &    5.700 &  5.100  & 5.250   &  5.150   &   5.200  \\ \hline
average Kendall-Tau distance      &   4.170 & 4.014   &   4.316 &  3.997   & 4.207    \\ \hline
unanimity width       & 2.500   &   2.500 & 2.700   &  2.400   & 2.350    \\ \hline
blocking size         &  3.350  & 3.350   &  3.500  &  3.300   &  3.300   \\ \hline
distance to consensus &   14.400 &  13.550  & 14.400   &   13.350  &  14.050  
\end{tabular}
\caption{The values of the parameters we consider for $\theta=1.5$.}
\label{table:theta-A}
\end{table}

\def\arraystretch{1.25}
\begin{table}[htbp]
\begin{tabular}{llllll}
\#Voters        & 10 & 25 & 50 & 100 & 200 \\ \hline
maximum range                 & 3.050   & 2.950    &    2.850 &  2.850   &  2.500   \\ \hline
average Kendall-Tau distance      & 1.028   & 0.883   &  0.863  &  0.896   &   0.741  \\ \hline
unanimity width       &  1.300  &  1.150  &  1.150  & 1.050    & 0.950    \\ \hline
blocking size         &    2.300 &  2.150  & 2.150   &    2.050 &   1.950  \\ \hline
distance to consensus &   4.650 &  3.950  &  3.750  & 3.850    &  3.000  
\end{tabular}
\caption{The values of the parameters we consider for $\theta=3$.}
\label{table:theta-B}
\end{table}

\def\arraystretch{1.25}
\begin{table}[htbp]
\begin{tabular}{llllll}
 \#Voters         & 10 & 25 & 50 & 100 & 200 \\ \hline
maximum range                 & 1.650   &   1.700 &  1.600  & 1.650    &   1.500  \\ \hline
average Kendall-Tau distance      &  0.130  &   0.180 & 0.166   &  0.168   &   0.130  \\ \hline
unanimity width       &   0.650 &    0.550 &  0.600  &  0.650   &   0.500  \\ \hline
blocking size         &  1.650  &  1.550  & 1.600   & 1.650    &  1.500   \\ \hline
distance to consensus &   0.650 &  0.900  &  0.750  &   0.800  &    0.650
\end{tabular}
\caption{The values of the parameters we consider for $\theta=5$.}
\label{table:theta-C}
\end{table}

\def\arraystretch{1.25}
\begin{table}[htbp]
\begin{tabular}{llllll}
\#Voters        & 10 & 25 & 50 & 100 & 200 \\ \hline
maximum range                 &    1.000 & 1.050   & 1.050   &  1.150   & 1.000    \\ \hline
average Kendall-Tau distance      & 0.000   & 0.010   &  0.010  &  0.030   &  0.000   \\ \hline
unanimity width       &  0.000  & 0.050   &   0.050 &  0.150   &  0.000   \\ \hline
blocking size         & 1.000   &  1.050  &  1.050  &   1.150  &  1.000   \\ \hline
distance to consensus &   0.000 &  0.050  &  0.050  &     0.150 &   0.000 
\end{tabular}
\caption{The values of the parameters we consider for $\theta=8$.}
\label{table:theta-D}
\end{table}

\section{Algorithms for \DisKRA}

We start with an easy Turing reduction from \DisOPTKRA to \DisKRA.

\begin{observation}\label{obs:gen_OPT}
Suppose there exists an algorithm for \DisKRA running in time $\OO(f(m,n))$. Then there exists an algorithm for \DisOPTKRA running in time $\OO(f(m,n)\log(mn))$.
\end{observation}

\begin{proof}
We note that the optimal Kemeny score belongs to the set $\{0,1,\ldots,n{m\choose 2}\}$. To solve \DisOPTKRA, we perform a binary search in the range from $0$ to $n{m\choose 2}$ to find the smallest $k$ such that the algorithm for \DisKRA returns at least one ranking.
\end{proof}

We now present a bounded search based \FPT algorithm for \DisKRA parameterized by the optimal Kemeny score. Hence, we also have an \FPT algorithm for \DisOPTKRA parameterized by the optimal Kemeny score.

\begin{theorem}\label{fpt:k}
 Let $k$ be the Kemeny score of a Kemeny ranking. There is an \FPT algorithm for \DisKRA parameterized by $k$ which runs in time $\OO^*\left(2^{k}\right)$. Hence, we have an \FPT algorithm for \DisOPTKRA parameterized by $k_{\text{OPT}}$ which runs in time $\OO^*\left(2^{k_{\text{OPT}}}\right)$.
\end{theorem}
\begin{proof}
Due to \Cref{obs:gen_OPT}, it is enough to present an algorithm for \DisKRA.
We design an algorithm for a more general problem \DisKRAprime where every output ranking needs to respect the relative order of some set of pair of candidates given as input. If the set of pairs of candidates is empty, then the new problem is the same as \DisKRA.

Let $(\CC,\Pi,k,r)$ be an arbitrary instance of \DisKRA. We define $\XX=\{a>b: a,b\in \CC,\text{ every ranking in }\Pi\text{ prefers }a\text{ over }b\}$ to be the unanimity order of $\Pi$. We find a solution of \DisKRAprime instance $(\CC,\Pi,k,r,\XX)$. We now design a bounded search based algorithm. We maintain a set \SS of solutions, which is initialized to the empty set. If every pair of candidates belong to \XX and $k\ge0$, then we put the ranking induced by \XX in \SS. If $k<0$, then we discard this branch. Otherwise, we pick a pair $(a,b)$ of candidates not present in \XX, solve $(\CC,\Pi,k-|\{\pi\in\Pi: ~b \succ a \text{ in }\pi \}|,r,\text{transitive closure of }\XX\cup\{a>b\})$ and $(\CC,\Pi,k-|\{\pi\in\Pi: ~a \succ b \text{ in }\pi\}|,r,\text{transitive closure of }\XX\cup\{b>a\})$ recursively, and put solutions found in \SS. We note that, since $(a,b)$ is not a unanimous order of $\Pi$, the target Kemeny score $k$ decreases by at least one on both the branches of the search tree. Hence, the height of the search tree is at most $k$. Thus, the number of leaves and nodes in the search tree are at most respectively $2^k$ and $2\cdot 2^k$. After the search terminates, we output $\min\{r,|\SS|\}$ rankings from \SS. If \SS remains empty set, report that there is no ranking whose Kemeny score is at most $k$. The computation at each node of the search tree (except the recursive calls) clearly takes a polynomial amount of time. Hence, the runtime of our algorithm is $\OO^*\left(2^{k}\right)$. The correctness of our algorithm follows from the observation that every ranking $R$ whose Kemeny score is at most $k$, appears in a leaf node of the search tree of our algorithm. This also follows from Section 4.2 of \cite{DBLP:journals/tcs/BetzlerFGNR09}.
\end{proof}

Running the algorithm in \Cref{fpt:k} with target Kemeny score $\lambda k$ where $k$ is the optimal Kemeny score gives us the following result.

\begin{corollary}\label{cor:fpt-k}
There is an algorithm for \DisApproxKRA running in time $\OO^*\left(2^{\lambda k}\right)$ parameterized by both $\lambda \text{ and }k$.
\end{corollary}

We now consider the number of candidates $m$ as our parameter and present a dynamic programming based \FPT algorithm for \DisKRA.

\begin{theorem}\label{fpt:m}
There is an algorithm for \DisKRA which runs in time $\OO^*\left(2^m r^{\mathcal{O}(1)}\right)$. In particular, \DisKRA and \DisOPTKRA are \FPT parameterized by the number of candidates since the number $r$ of output rankings can be at most $m!$.
\end{theorem}

\begin{proof}
Let $(\CC,\Pi,k,r)$ be an arbitrary instance of \DisKRA. We maintain a dynamic programming table \TT indexed by the set of all possible non-empty subsets of \CC. For a subset $\SS\subseteq \CC, \SS\neq\emptyset$, the table entry $\TT[\SS]$ stores at most $\min\{r,|\SS|!\}$ distinct rankings on \SS which have the least Kemeny score when the votes are restricted to \SS. Let us define $\kappa= \min\{r,|\SS|!\}$. We initialize table \TT for the trivial cases like $\TT[\SS] = \left(  \right) \text{ when } |\SS| = 0, ~\TT[\SS] = \left( \text{the element from } \SS \right) \text{ when } |\SS| = 1 \text{ and } \TT[\SS] = \left( x \succ y \right) \text{ when } \SS = \left\{ x, y \right\}$ and $ x \succ y $ has the least Kemeny score when $\Pi$ is restricted to $\left\{ x, y \right\}$ or $ \TT[\SS] = \left( x \succ y, ~y \succ x \right) \text{ when } \SS = \left\{ x, y \right\}$ and both $ x \succ y  $ and $y \succ x$ have the least Kemeny score when $\Pi$ is restricted to $\left\{ x, y \right\}$.
% To update the table entry $\TT[\SS]$, we pick any candidate $c\in\SS$ and ``guess" its positions in the $\min\{r,|\SS|!\}$ distinct rankings which the table entry $\TT[\SS]$ is supposed to store.
% To update the table entry $\TT[\SS]$, we ``guess'' the first candidates $c_1, \ldots,c_\kappa$ (they need not be distinct) of the $\kappa$ distinct rankings which the table entry $\TT[\SS]$ is supposed to store and ``fetch'' the best rankings from the table entries $\TT[\SS\setminus\{c_1\}],\ldots,\TT[\SS\setminus\{c_r\}]$.
To update the table entry $\TT[\SS]$ for $|\SS| \ge 3$, we include to that entry $\min\{r, |\SS|!\}$ rankings that have the least Kemeny score (when the votes are restricted to $\SS$) among all rankings of the form $c>\pi$, where $c$ is a candidate in $\SS$ and $\pi$ is a ranking stored in $\TT[\SS\setminus\{c\}]$. Updating each table entry takes at most $\mathcal{O}^{\star}(r^{\mathcal{O}(1)})$ time. As there are $2^{m}-1$ table entries, the running time of our algorithm is at most $\mathcal{O}^{\star}\big(2^m r^{\mathcal{O}(1)}\big)$.
% Updating each table entry takes $\OO^*\left({m+r-1\choose r-1}\right)=\OO^*\left(m^r\right)$. Since the table size is $2^m-1$, the running time of our algorithm is $\OO^*\left(2^m m^r\right)$.

We now present the proof of correctness of our algorithm. Suppose we have $\SS=\{c_1,...,c_{\el}\}$ and $c_1>...>c_{\el}$ be a ranking in $\TT[\SS]$. Then $c_1>...>c_{\el}$ is a Kemeny ranking if the votes in $\Pi$ are restricted to $\SS$. But then $c_2>...>c_{\el}$ is a Kemeny ranking if votes are restricted to $\SS \setminus \{c_1\}$. If not, then suppose $c_2^{\prime}>...>c_{\el}^{\prime}$ be a ranking with Kemeny score less than $c_2>...>c_{\el}$. Then the Kemeny score of $c_1>c_2^{\prime}>...>c_{\el}^{\prime}$ is less than the Kemeny score of $c_1>c_2>...>c_{\el}$ contradicting our assumption that $c_1>...>c_{\el}$ is a Kemeny ranking when votes are restricted to \SS. Hence, the update procedure of our dynamic programming algorithm is correct.
\end{proof}

% \begin{theorem}
% There is an algorithm for \DKRA which runs in time $\OO^*\left(1.53^{kr}\right)$. In particular, \DKRA is \FPT parameterized by the Kemeny score and the number $r$ of output rankings.
% \end{theorem}

% \begin{proof}
% Let $(C,\Pi,k,r,d,s)$ be an arbitrary instance of \DKRA. Betzler et al.~\cite{DBLP:journals/tcs/BetzlerFGNR09} discovered a bounded search tree based algorithm for finding a ranking with Kemeny score at most $k$ which runs in time $\OO^*(1.53^k)$. That search tree has $\OO^*(1.53^k)$ many leaf nodes and each leaf node corresponds to a plausible ranking whose Kemeny score could be at most $k$. Moreover, a critical observation reveals that every ranking whose Kemeny score is at most $k$ appears in at least one leaf node of that search tree. The algorithm by Betzler et al. terminates when it finds any ranking whose Kemeny score is at most $k$. We instead traverse all the leaf nodes of that search tree and first construct the set \XX of the set of all rankings whose Kemeny score is at most $k$. Since the search tree has $\OO^*(1.53^k)$ leaf nodes, it follows that $|\XX|=\OO^*(1.53^k)$. We go over all subsets of \XX of size $r$ to find if there exists any such subset whose Kendall-Tau diversity is at least $d$ and scatteredness is at least $s$. Hence, the running time of our algorithm is $\OO^*({1.53^k\choose r})=\OO^*(1.53^{kr})$.
% \end{proof}
\Cref{corollary:Dis_Approx_fpt_m} follows immediately from the algorithm presented in the proof of \Cref{fpt:m}.

\begin{corollary}
\DisApproxKRA is FPT parameterized by the number of candidates $m$.
\label{corollary:Dis_Approx_fpt_m}
\end{corollary}
\begin{proof}
Consider an instance $(\CC, \Pi, \lambda, r)$ of \DisApproxKRA. We run the algorithm of Theorem \ref{fpt:m} on instances $(\CC,\Pi,0,1), (\CC,\Pi,1,1)$, $\ldots$ of \DisKRA. We stop once we encounter a YES instance, say $(\CC,\Pi,k^{*},1)$. Note that $k^{*}$ is the optimum Kemeny score for the election profile $(\CC,\Pi)$. Next, we run the algorithm of Theorem \ref{fpt:m} on the instance $(\CC,\Pi,\lambda\cdot k^{*},r)$ of \DisKRA to get the desired output. As $k^{*}\leq \binom{m}{2}\cdot |\Pi|$, the overall running time of the algorithm is at most $\mathcal{O}^{*}\big(2^m r^{\mathcal{O}(1)}\big)$. So, as $r\leq m!$, it follows that \DisApproxKRA is FPT parameterized by the number of candidates $m$.
\end{proof}

Our next parameter is the ``average pairwise distance (Kendall-Tau distance)'' $d$ of the input rankings. We present a dynamic programming based \FPT algorithm parameterized by $d$.

\begin{theorem}\label{fpt:d} Let $d$ be the average KT-distance of an election $\left( \Pi, \CC \right)$. There is an \FPT for \DisOPTKRA parameterized by $d$ which runs in time $\OO^{\star}\left( 16^{d} \right)$.
\end{theorem}
\begin{proof} Let $|\CC| = m, ~|\Pi| = n \text{ and } p_{avg} \left( c \right) \coloneqq \frac{1}{n} \cdot \sum\limits_{v \in \Pi}v(c)$ where $v(c) \coloneqq |\left\{ c^{\prime} \in \CC : c^{\prime} \succ c \text{ in }v \in \Pi \right\}|$. Formally for an election $\left(\Pi, \CC\right)$, $d \coloneqq \frac{\sum\limits_{v \in \Pi} \sum\limits_{w \in \Pi}\text{d}_{\text{KT}}\left( v, w \right)}{n \cdot (n - 1)}$. Following the proof of both Lemma 6 and Lemma 7 from Betzler et al. \cite{DBLP:journals/tcs/BetzlerFGNR09}, we have a set of candidates say ${P_{i}} \coloneqq \left\{ c \in \CC \mid p_{avg}(c) - d < i < p_{avg}(c) + d \right\}$ for each position $i \in \left[ m - 1 \right]_{0}$ in an optimal Kemeny Consensus and we know that $|P_{i}| \leqslant 4d ~~\forall i\in \left[ m - 1 \right]_{0}$. Our FPT dynamic programming algorithm is an extension of the algorithm presented in Fig. 4. of section 6.4 of \cite{DBLP:journals/tcs/BetzlerFGNR09}. 

Let the subset of candidates that are forgotten at latest at position $i$, be denoted by $F(i) \coloneqq P_{i-1} \setminus P_{i}$ and  the subset of candidates that are introduced for the first time at position $i$ be denoted by $I(i) \coloneqq P_i \setminus P_{i-1}$. We maintain a three dimensional dynamic programming table \TT indexed by $\forall i \in \left[ m - 1\right]_{0}, \forall \text{ }c \in P_{i} \text{ and } \forall P_{i}^{\prime} \subseteq P_{i} \setminus \left\{ c \right\}$ of size at most $\OO\left( 16^{d} \cdot d \cdot m \right)$. We define the partial Kemeny Score $\text{pK-score}(c,\RR) \coloneqq \sum\limits_{c^{\prime} \in \RR}\sum\limits_{v \in \Pi}d_{v}^{\RR}(c, c^{\prime})$ where $ d_{v}^{\RR}(c, c^{\prime}) \coloneqq  0 \text{ if } c \succ_{v} c^{\prime}  \text{ and }d_{v}^{\RR}(c, c^{\prime}) \coloneqq 1 \text{ otherwise}$ and $\RR \subseteq \CC$.
At each table entry $\TT(i, c, P_{i}^{\prime})$, we store a sequence of at most $\min\left(r, 4d \right)$ number of partial Kemeny Scores sorted in non-decreasing order by considering and iterating over the entries in $\TT(i-1, c^{\prime}, (P_{i}^{\prime} \cup F(i)) \setminus \left\{ c^{\prime} \right\} )$ $\forall c^{\prime} \in P_{i}^{\prime} \cup F(i)$ and we store the tuple \[ \biggl( \TT(i-1, c^{\prime}, (P_{i}^{\prime} \cup F(i)) \setminus \left\{ c^{\prime} \right\} ) \] \[ + \text{pK-score}(c, (P_{i} \cup \bigcup\limits_{i<j<m}I(j)) \setminus (P_{i}^{\prime} \cup \left\{ c \right\})) \biggl)_{c^{\prime} \in P_{i}^{\prime} \cup F(i)}\] in that table entry unlike storing only the minimum partial Kemeny Score at each table entry. K-score of an election is the Kemeny Score of an optimal Kemeny ranking. $\text{K-score}(\Pi, \CC) = \sum\limits_{i = 0}^{m-2} \text{pK-score}(c_i, \RR_i)$.

At each entry of the table candidate $c$ takes position $i$ and all of $P_{i}^{\prime}$ take position smaller than $i$. The initialization step is same as the algorithm presented in Fig. 4. of section 6.4 of \cite{DBLP:journals/tcs/BetzlerFGNR09} but the difference lies in the update step of that algorithm. Though we are storing Kemeny score in each table entry, we can enumerate Kemeny ranking(s) from them within asymptotic bound of our current run time by iteratively ordering the candidate(s) for which we get minimum partial Kemeny Score in a particular table entry. We Output first $r$ number of optimal Kemeny rankings whose K-scores are stored in the entry $T(m-1, c, P_{m-1} \setminus \left\{ c \right\})$ where $r \leqslant 4d \leqslant 4m^{2} << m!$. Correctness of Lemma 8 of \cite{DBLP:journals/tcs/BetzlerFGNR09} ensures the correctness of our algorithm for generating at most  $min\left(r, 4d \right)$ number of optimal Kemeny Rankings.

Updating each table entry takes time at most $\min(r, 4d) \cdot (4d + nm \log m)$ time. Hence, the overall runtime is bounded above by $\OO^{\star} \left( 16^{d} \right)$.
\end{proof}

We next consider the ``maximum range" $r_{max}$ of candidate positions in the input rankings, as our parameter. We again present a dynamic programming based \FPT algorithm parameterized by $r_{max}$.

\begin{theorem}\label{fpt:rmax} Let $r_{max}$ be the maximum candidate position range of an election $\left( \Pi, \CC \right)$. There exists an \FPT dynamic programming algorithm for \DisOPTKRA parameterized by $r_{max}$ which runs in time $\mathcal{O}^{*}\left( 32^{r_{max}} \right)$.
\end{theorem}
\begin{proof}
Following the proof of both Lemma 9 and Lemma 10 from \cite{DBLP:journals/tcs/BetzlerFGNR09}, we have here $|P_{i}| \leqslant 6r_{max}$. We maintain a dynamic programming table $\TT$ of size $\OO \left( 32^{r_{max}} \cdot r_{max} \cdot m \right)$ indexed by $\forall i \in \left[ m - 1\right]_{0}, \forall \text{ }c \in P_{i} \text{ and } \forall P_{i}^{\prime} \subseteq P_{i} \setminus \left\{ c \right\}$. The proof of Theorem \ref{fpt:rmax} follows immediately from a complete analogy to the proof of Theorem \ref{fpt:d}.
\end{proof} 

Our final parameter is the unanimity width of the input rankings. We present a dynamic programming based \FPT algorithm.

\begin{theorem}\label{unanimitywidthalgo}
$\textsc{Distinct OPT Kemeny Rank Aggregation}$ admits an \FPT algorithm in the combined parameter unanimity width $w$ and number of rankings $r$, which runs in time $\mathcal{O}^{*}\big(2^{\mathcal{O}(w)}\cdot r\big)$.
\end{theorem}
\begin{proof}
The problem of finding a Kemeny consensus is known to admit an FPT algorithm in the parameter $w$ (Section 3, \cite{arrighi2020width}). We adapt this algorithm to prove~\Cref{unanimitywidthalgo}. Consider an instance $(\CC,\Pi,r)$ of \DisOPTKRA. Let $m$ denote the number of candidates in $\CC$, and let $n$ denote the number of voters in $\Pi$. For any candidates $a,b\in \CC$, let  $cost(a,b)$ denote the number of voters in $\Pi$ who prefer $b$ over $a$. Note that for any linear ordering $\pi$ of candidates, $\text{Kemeny}_{\Pi}(\pi) = \sum_{a,b \in \CC: a \succ b \text{ in }\pi}cost(a, b)$. 
% that extends $\rho$ 
% $$\mbox{Kemeny score of } \pi = \underset{\substack{a,b\in C:\\a \mbox{ appears before }b \mbox{ in } \pi}}{\sum}cost(a,b)$$ 
Let $\rho$ denote the unanimity order of $\Pi$. Let $G_{\rho}$ denote the cocomparability graph of $\rho$. Using Lemma 3 of \cite{DBLP:conf/ijcai/ArrighiFLO021}, let's construct a nice $\rho$-consistent path decomposition, say $\mathcal{P}=(B_1,\ldots,B_{2m}),$ of $G_{\rho}$ of width $w'\leq 5w+4$ in time $\mathcal{O}\big(2^{\mathcal{O}(w)}\cdot m\big)$.\\
$~$\\For each $1\leq i\leq 2m$, 
\begin{itemize}
\item Let $forg(i)$ denote the set of candidates that have been forgotten up to $i^{th}$ bag. That is, $forg(i) = \big(B_1\cup\ldots\cup B_{i-1}\big)\setminus B_i$.
\item For each candidate $v\in B_i$, let $\mathcal{A}(i,v)$ denote the cost incurred by the virtue of placing all candidates 
of $forg(i)$ before $v$. That is, $\mathcal{A}(i,v)=\underset{u\in forg(i)} {\sum}cost(u,v)$.
\item For each candidate $v\in B_i$ and each $T\subseteq B_i\setminus \{v\}$, let $\mathcal{B}(i,v, T)$ denote the cost incurred by the virtue of placing all candidates of $T$ before $v$. That is, $\mathcal{B}(i,v,T) = \underset{u\in T}{\sum}cost(u,v)$.
\item For each $T\subseteq B_i$, let $C(i,T)$ be a set that consists of first $min\big(r,|forg(i)\uplus T|!\big)$ orderings, along with their Kemeny scores, if all linear extensions of $\rho$ on $forg(i) \uplus T$ were to be sorted in ascending order of their Kemeny scores. That is, $C(i,T)$ consists of the tuples $(\pi_1,k_1), (\pi_2,k_2),\ldots$, where $\pi_1,\pi_2,\ldots$ are the first $min\big(r,|forg(i)\uplus T|!\big)$ orderings in the sorted order, and $k_1,k_2,\ldots$ are their respective Kemeny scores. 
\end{itemize}
Recall that every Kemeny consensus extends $\rho$ (Lemma 1, \cite{DBLP:journals/tcs/BetzlerFGNR09}). So, if all linear extensions of $\rho$ on $\CC$ were to be sorted in ascending order of their Kemeny scores, then all Kemeny consensuses would appear in the beginning. Thus, $(\CC,\Pi,r)$ is a YES instance if and only if $C(2m,\phi)$ contains $r$ orderings of the same Kemeny score.\\
$~$\\
Let's use DP to find all $\mathcal{A}(\cdot,\cdot)$'s, $\mathcal{B}(\cdot,\cdot,\cdot)$'s and $C(\cdot,\cdot)$'s as follows:
\begin{itemize}
\item First, let's compute and store $\mathcal{A}(i,\cdot)$'s in a table for $i=1,\ldots,2m$ (in that order) in time $\mathcal{O}\big(w'\cdot m\cdot \log(m\cdot n)\big)$ as follows: We set $\mathcal{A}(1,u)=0$, where $u$ denotes the candidate introduced by $B_1$. Now, consider $i\geq 2$ and a candidate $v\in B_i$. Let's describe how to find $\mathcal{A}(i,v)$.\\
$~$\\
\textbf{Introduce node}.\\ Suppose that $B_i$ introduces a candidate, say $x$. Note that $forg(i) = forg(i-1)$. So, if $v\neq x$, we set $\mathcal{A}(i,v) = \mathcal{A}(i-1,v)$. Now, suppose that $v=x$. Let's show that $cost(u,x)=0$ for all $u\in forg(i)$. Consider   $u\in forg(i)$. In $\mathcal{P}$, $u$ is forgotten before $x$ is introduced. So, $\{u,x\}\not\in E(G_{\rho})$. That is, $u$ and $x$ are comparable in $\rho$. Also, due to $\rho$-consistency of $\mathcal{P}$, we have $(x,u)\not\in\rho$. Therefore, $(u,x)\in \rho$. That is, all voters in $\Pi$ prefer $u$ over $x$. So, $cost(u,x)=0$. Thus, we set $\mathcal{A}(i,x)=0$.\\
$~$\\
\textbf{Forget node}.\\ Suppose that $B_i$ forgets a candidate, say $x$. Note that $forg(i) = forg(i-1) \uplus \{x\}$. So, we set $\mathcal{A}(i,v) = \mathcal{A}(i-1,v) + cost(x,v)$.\\
\item Next, let's compute and store all $\mathcal{B}(\cdot,\cdot,\cdot)$'s in a table in time $\mathcal{O}\big(w'\cdot 2^{w'}\cdot m\cdot \log(m\cdot n)\big)$ as follows: Consider $1\leq i\leq 2m$ and $v\in B_i$. We have $\mathcal{B}(i,v,\phi)=0$. Let's set $\mathcal{B}(i,v,T)$ for non-empty subsets $T\subseteq B_i\setminus\{v\}$ (in ascending order of their sizes) as $\mathcal{B}(i,v,T\setminus\{u\}) + cost(u,v)$, where $u$ denotes an arbitrary candidate in $T$.\\
\item Next, let's compute and store $C(i,\cdot)$'s in a table in time $\mathcal{O}\big(w'\cdot 2^{w'}\cdot m^2\cdot r\cdot \log(m\cdot n\cdot r)\big)$ for $i=1,\ldots, 2m$ (in that order) as follows: We set $C(1,\phi) = \{(,0)\}$ and $C(1,\{u\}) = \{(u,0)\}$, where $u$ denotes the candidate introduced by $B_1$. Now, consider $i\geq 2$. Let's describe how to find $C(i,\cdot)$'s.\\
$~$\\
\textbf{Introduce node}.\\
Suppose that $B_i$ introduces a candidate, say $x$. For each $T\subseteq B_i$ that does not contain $x$, we set $C(i,T) = C(i-1,T)$.\\$~$\\
Now, let's find $C(i,T)$ for all subsets $T\subseteq B_i$ that contain $x$ (in ascending order of their sizes)  as follows: First, let's consider $T=\{x\}$. Recall that $(u,x)\in\rho$ for all $u\in forg(i)$. So, $x$ is the last candidate in all linear extensions of $\rho$ on $forg(i)\uplus \{x\}$. Also, in any such ordering, the pairs of the form $(u,x)$, where $u\in forg(i)$, contribute $0$ to Kemeny score. Thus, we put the tuples $\big(\pi_1>x,s_1\big),\big(\pi_2>x,s_2\big),\ldots$ in $C(i,\{x\})$, where $(\pi_1,s_1),(\pi_2,s_2),\ldots$ denote the tuples of $C(i-1,\phi)$, and $\pi_1>x,\pi_2>x,\ldots$ denote the orderings obtained by appending  $x$ to $\pi_1,\pi_2,\ldots$ respectively.\\
$~$\\
Now, let's consider a subset $T\subseteq B_i$ of size $\geq 2$ that contains $x$. Let's describe how to find $C(i,T)$. Let $\Delta(i,T)$ denote the set of all candidates $c\in T$ such that $c$ is not unanimously preferred over any other candidate of $forg(i)\uplus T$. That is, there is no other candidate $u\in forg(i)\uplus T$ such that $(c,u)\in \rho$. Recall that $x$ appears after all candidates of $forg(i)$ in any linear extension of $\rho$ on $forg(i)\uplus T$. So, it is clear that in any such ordering, the last candidate (say $y$) belongs to $\Delta(i,T)$. Moreover,
\begin{itemize}
\item The pairs of the form $(u,y)$, where $u\in forg(i)$, together contribute $\mathcal{A}(i,y)$ to Kemeny score.
\item The pairs of the form $(u,y)$, where $u\in T\setminus \{y\}$, together contribute $\mathcal{B}(i,y,T\setminus \{y\})$ to Kemeny score.
\end{itemize}
So, to find $C(i,T)$,  let's proceed as follows: We compute $\Delta(i,T)$. For each possible choice $y\in \Delta(i,T)$ of the last candidate, let's form a set, say $\Gamma(y)$, that consists of the following tuples:
\begin{itemize}
\item $\Big(\pi_1^{y}>y, s_1^y+\mathcal{A}(i,y)+\mathcal{B}\big(i,y,T\setminus\{y\}\big)\Big)$
\item $\Big(\pi_2^{y}>y, s_2^y+\mathcal{A}(i,y)+\mathcal{B}\big(i,y,T\setminus\{y\}\big)\Big)$ and so on
\end{itemize}
where $(\pi_1^y, s_1^y), (\pi_2^y, s_2^y),\ldots$ denote the tuples of $C(i, T\setminus\{y\})$, and $\pi_1^y>y, \pi_2^y>y,\ldots$ denote the orderings obtained by appending $y$ to $\pi_1^y,\pi_2^y,\ldots$ respectively. Finally, let's sort all tuples of $\underset{y\in \Delta(i,T)}{\biguplus}\Gamma(y)$ in ascending order of their Kemeny scores, and put the first $min\big(r, |forg(i)\uplus T|!\big)$ of them in $C(i,T)$.
\\$~$\\
\textbf{Forget node}.\\
Suppose that $B_i$ forgets a candidate, say $x$. For each $T\subseteq B_i$, as  $forg(i)\uplus T = forg(i-1)\uplus \big(T\uplus\{x\}\big)$, we set $C(i,T) = C(i-1,T\uplus\{x\})$.
\end{itemize}
This concludes the proof of ~\Cref{unanimitywidthalgo}.
\end{proof}
\begin{corollary}\DisApproxKRA is FPT in the combined parameter unanimity width $w$ and number of rankings $r$. 
\label{corollary:fpt_unanimity_width_r}
\end{corollary}
\begin{proof}
Consider an instance \DisApproxKRA. As in the algorithm described in the proof of \Cref{unanimitywidthalgo}, we find all $\mathcal{A}(\cdot,\cdot)$'s, $\mathcal{B}(\cdot,\cdot,\cdot)$'s and  $C(\cdot, \cdot)$'s. Note that $(\CC,\Pi,\lambda,r)$ is a YES instance if and only if $C(2m,\phi)$ contains $r$ orderings, and the Kemeny score of the $r^{th}$ ordering is at most $\lambda$ times the Kemeny score of the first ordering. The overall running time of the algorithm is at most $\mathcal{O}^{*}\big(2^{\mathcal{O}(w)}\cdot r\big)$. This proves \Cref{corollary:fpt_unanimity_width_r}.
\end{proof}

Our last result is an \FPT algorithm for \DisApproxKRA parameterized by the average Kendall-Tau distance $d$ and the approximation parameter $\lambda$.
Here we aim to relate the position of a candidate $c$ in a $\lambda$-approximate ranking $\pi$, \textit{i.e.} a ranking whose Kemeny Score denoted by $\text{K-score}\left( \pi \right)$ has value at most $\lambda \cdot k_{OPT} \text{ where } k_{OPT}$ denotes the optimal Kemeny Score, to its average position in the set of votes $\Pi$ denoted by $p_{avg}(c)$. Our resulting \Cref{lemma:davg_pavg_lambda_approx} depends on Lemma 6 from \cite{DBLP:journals/tcs/BetzlerFGNR09}.  

\begin{lemma}[$\star$]
\label{lemma:davg_pavg_lambda_approx}
 $p_{avg}(c) - \lambda \cdot d \le \pi(c) \le p_{avg}(c) + \lambda \cdot d$ where $\pi(c)$ denotes position of $c$ in $\pi$ and $d$ is average KT-distance.
\end{lemma}
\begin{proof}
There can be two cases for a vote $v \in \Pi$.
\begin{case}
 $v(c) \le \pi(c)$
 \label{case:v(c)_le_pi(c)}
\end{case}
In \Cref{case:v(c)_le_pi(c)} there are $\pi(c) - 1$ candidates that appear before $c$ in $\pi$. Note that at most $v(c) - 1$ of them can appear before $c$ in $v$. Hence, at least $\pi(c) - v(c)$ of them must appear after $c$ in $v$. Thus, $\text{d}_{\text{KT}}\left( v, \pi \right) \ge \pi(c) - v(c)$.
\begin{case}
 $v(c) > \pi(c)$
 \label{case:v(c)_grt_pi(c)}
\end{case}
Here in \Cref{case:v(c)_grt_pi(c)}, we come up with $\text{d}_{\text{KT}}\left( v, \pi \right) \ge v(c) - \pi(c)$ arguing similarly to \Cref{case:v(c)_le_pi(c)}.
 \begin{align} &\text{K-score}\left( \pi \right) = \sum\limits_{v \in \Pi} \text{d}_{\text{KT}}(v, \pi)\nonumber \\
&= \sum\limits_{v \in \Pi : v(c) \le \pi(c)} \text{d}_{\text{KT}}(v, \pi) + \sum\limits_{v \in \Pi : v(c) > \pi(c)} \text{d}_{\text{KT}}(v, \pi) \nonumber \\
&\ge \sum\limits_{\substack{v \in \Pi : \\ v(c) \le \pi(c)}} \left(\pi(c) - v(c)\right) + \sum\limits_{\substack{v \in \Pi : \\v(c) > \pi(c)}} \left( v(c) - \pi(c) \right)~~~ \left[ \text{using } \Cref{case:v(c)_le_pi(c)} \text{ and }\Cref{case:v(c)_grt_pi(c)} \right] \label{eqn:1_K-score-pi_ge_pi_c-v_c}
\end{align}
Note that
\begin{align}
    &\sum\limits_{v \in \Pi : v(c) \le \pi(c)} \left( \pi(c) - v(c) \right) + \sum\limits_{v \in \Pi : v(c) > \pi(c)} \left( v(c) - \pi(c) \right) \nonumber \\
    &= \sum\limits_{v \in \Pi}v(c) - 2 \sum\limits_{\substack{v \in \Pi : \\ v(c) \le \pi(c)}} v(c) + \pi(c) \cdot \left( 2 \cdot |\left\{ v \in \Pi : v(c) \le \pi(c) \right\}| - n \right) \nonumber \\
    &= n \cdot p_{avg}(c) - n \pi(c) - 2 \sum\limits_{\substack{v \in \Pi : \\ v(c) \le \pi(c)}} v(c) + \pi(c) \cdot \left( 2 \cdot |\left\{ v \in \Pi : v(c) \le \pi(c) \right\}| \right) \nonumber \\
    &\ge n\left( p_{avg}(c) - \pi(c) \right) \label{eqn:2_sum_pi_c-v_c_ge_n_p_avg_c-pi_c}
\end{align}
Similarly,
\begin{align}
    &\sum\limits_{v \in \Pi : v(c) \le \pi(c)} \left( \pi(c) - v(c) \right) + \sum\limits_{v \in \Pi : v(c) > \pi(c)} \left( v(c) - \pi(c) \right) \nonumber \\
    &= - \sum\limits_{v \in \Pi}v(c) + 2 \sum\limits_{\substack{v \in \Pi : \\ v(c) > \pi(c)}} v(c) + \pi(c) \cdot \left( - 2 \cdot |\left\{ v \in \Pi : v(c) > \pi(c) \right\}| + n \right) \nonumber \\
    &= -n \cdot p_{avg}(c) + n \pi(c) + 2 \sum\limits_{\substack{v \in \Pi : \\ v(c) > \pi(c)}} v(c) - \pi(c) \cdot \left( 2 \cdot |\left\{ v \in \Pi : v(c) > \pi(c) \right\}| \right) \nonumber \\
    &\ge -n\left( p_{avg}(c) - \pi(c) \right) \label{eqn:3_sum_pi_c-v_c_ge_-n_p_avg_c-pi_c}
\end{align}
Now let's show that
\begin{equation}
\label{eqn:4_pi_le_lambda_n_d}
\text{K-score}\left( \pi \right) \le \lambda \cdot n \cdot d
\end{equation}
We have
\begin{align}
 d &= \frac{\sum\limits_{v \in \Pi} \sum\limits_{w \in \Pi}\text{d}_{\text{KT}}\left( v, w \right)}{n \cdot (n - 1)} \nonumber \\
 d &\ge \frac{ n \cdot \sum\limits_{w \in \Pi, w \neq v^{\star}} \text{d}_{\text{KT}} \left( v^{\star}, w \right)}{n \cdot \left( n - 1 \right)} > \frac{\sum\limits_{w \in \Pi, w \neq v^{\star}} \text{d}_{\text{KT}} \left( v^{\star}, w \right)}{n} \nonumber \\ &\left[ \exists v^{\star} \in \Pi \text{ for which } \sum\limits_{w \in \Pi, w \neq v^{\star}} \text{d}_{\text{KT}} \left( v^{\star}, w \right) \text{ is minimum } \right] \nonumber \\
 &\implies \text{ K-score}\left( v^{\star} \right) < n \cdot d \nonumber \\
&\text{So, } k_{OPT} \le \text{K-score}\left( v^{\star} \right) <n \cdot d \label{eqn:k_OPT_n_d} \\
& \text{K-score}\left( \pi \right) \le \lambda \cdot k_{OPT} < \lambda \cdot n \cdot d ~\left[ \text{Using \Cref{eqn:k_OPT_n_d}} \right] \label{eqn:proof_of_K_score_pi} \\
&\text{\Cref{eqn:proof_of_K_score_pi} proves \Cref{eqn:4_pi_le_lambda_n_d}} \nonumber \\
&\text{Now } \lambda \cdot n \cdot d \ge \text{K-score}\left( \pi \right) \ge n \cdot \left( p_{avg}(c) - \pi(c) \right) \nonumber \\ &\left[ \text{Using  \Cref{eqn:1_K-score-pi_ge_pi_c-v_c}, (\ref{eqn:2_sum_pi_c-v_c_ge_n_p_avg_c-pi_c}) \&  (\ref{eqn:4_pi_le_lambda_n_d}}) \right] \nonumber \\
&\implies p_{avg}(c) - \lambda \cdot d \le \pi(c) \label{eqn:approx_pi_c_part_1} \\
&\text{Again } \lambda \cdot n \cdot d \ge -n \cdot \left( p_{avg}(c) - \pi(c) \right) \nonumber \\ &\left[ \text{Using \Cref{eqn:1_K-score-pi_ge_pi_c-v_c}, (\ref{eqn:3_sum_pi_c-v_c_ge_-n_p_avg_c-pi_c}) \& (\ref{eqn:4_pi_le_lambda_n_d})} \right] \nonumber \\
&\implies \pi(c) \le p_{avg}(c) + \lambda \cdot d \label{eqn:approx_pi_c_part_2} \\
&\text{Hence } p_{avg}\left( c \right) - \lambda \cdot d \le \pi\left( c \right) \le p_{avg}\left( c \right) + \lambda \cdot d \label{eqn:approx_pi_c_complete} \\ &\left[ \text{Using \Cref{eqn:approx_pi_c_part_1} \& (\ref{eqn:approx_pi_c_part_2})} \right] \nonumber 
\end{align}
\Cref{eqn:approx_pi_c_complete} concludes the proof of \Cref{lemma:davg_pavg_lambda_approx}.
\end{proof}

The following \Cref{lemma:P_i_for_every_position_lambda_approx} depends on the Lemma 7 from \cite{DBLP:journals/tcs/BetzlerFGNR09}.

\begin{lemma}[$\star$]
\label{lemma:P_i_for_every_position_lambda_approx}
 $|P_{i}| \le 4\lambda d -1 ~~\forall i \in [m-1]_{0}$
\end{lemma}
\begin{proof}
We prove this lemma by contradiction. For this, we assume that for a position $i$, we have $|P_i| > 4 \lambda d$. Every candidate $c \in P_i$ has at most $2 \lambda d - 1$ different positions around its average position in a $\lambda$-approximately optimal Kemeny consensus $\pi$ based on the proof of  \Cref{lemma:davg_pavg_lambda_approx}. In \Cref{lemma:davg_pavg_lambda_approx} we have established that   $p_{avg}\left( c \right) - \lambda \cdot d < \pi(c) < p_{avg}\left( c \right) + \lambda \cdot d$. Hence, only those candidates in $\pi$ have $i$ as their common position for which
    \begin{align}
     |i - \pi(c) | &\le 2 \lambda d - 1\\
     \Rightarrow i - (2 \lambda d -1) &\le \pi(c) \le i + (2 \lambda d - 1) \label{i_2_lambda_d-1_pi_c_lambda_approx}
    \end{align}
     \par Since, our assumption is $|P_i| \ge 4 \lambda d$, therefore, each of these $4 \lambda d$ candidates must hold a position which differs at most by $2 \lambda d - 1$ around position $i$. But from \Cref{i_2_lambda_d-1_pi_c_lambda_approx}, we know that only those candidates in $\lambda$-approximately optimal Kemeny consensus $\pi$ qualify for $P_i$ whose $\pi(c)$ lies in the range $2 \lambda d - 1$ left and $2 \lambda d - 1$ right around position $i$. Therefore, we have $4 \lambda d - 1$ such positions. Hence, we approach towards contradiction.
\end{proof}

Having proved \Cref{lemma:davg_pavg_lambda_approx} we are now ready to design a dynamic programming algorithm in the following \Cref{fpt:approx_d} which has found its significant importance in this current setting.

\begin{theorem}\label{fpt:approx_d}
 There exists an FPT dynamic programming algorithm for \DisApproxKRA parameterized by both $\lambda \text{ and } d$ which runs in time $\mathcal{O} ^{*} ( 16^{\lambda d} )$.
\end{theorem}
\begin{proof}
 From the proof of \Cref{lemma:P_i_for_every_position_lambda_approx} and using similar arguments as in the proof of \Cref{fpt:d} we can conclude the proof of \Cref{fpt:approx_d}.
\end{proof}

\section{Concluding Remarks and Future Work}

We consider the problem of finding distinct rankings that have a good Kemeny score in either exact or approximate terms, and propose algorithms that are tractable for various natural parameterizations of the problem. We show that many optimal or close to optimal solutions can be computed without significant increase in the running time compared with the algorithms to output a single solution, which is in sharp contrast with the diverse version of the problem. We also establish a complete comparison between the five natural parameters associated with the problem, and demonstrate these relationships through experiments. 

We propose three main themes for future work. The first would be to extend these studies to other voting rules, and possibly identify meta theorems that apply to classes of voting rules. The second would be to understand if the structural parameters that we studied are correlated with some natural distance notion on the solution space: in other words, for a given distance notion, do all similar-looking instances have similar parameter values? Finally, we would also like to establish algorithmic lower bounds for the question of finding a set of diverse solutions that match the best known algorithms in the current literature.

\section*{Ethical Statement}

This paper adheres to the principles of research ethics, research integrity, and social responsibility. The research conducted in this paper is purely theoretical, and no human or animal subjects were involved in the research. Therefore, there were no ethical concerns related to the treatment of participants or the use of personal data.

We acknowledge that scientific research has the potential to impact society in significant ways, and we recognise our responsibility to consider the potential consequences of our research. We have taken steps to ensure that our research adheres to ethical and legal standards, and that it does not promote or contribute to harmful or unethical practices.

We have also taken measures to ensure the accuracy and reliability of our research results. We have followed established research protocols, and we have taken steps to prevent bias or conflicts of interest from influencing our findings.

Overall, we are committed to conducting research that is socially responsible and that contributes to the advancement of theoretical computer science in a responsible and ethical manner.

\bibliographystyle{alpha}
\bibliography{ArxivKemeny}
\end{document}